\newif\ifFull
\Fulltrue

\ifFull
  \documentclass[11pt]{article}
  \usepackage[margin=1in]{geometry}
  
  \usepackage{amsthm}
  \newtheorem{lemma}{Lemma}
  \newtheorem{theorem}{Theorem}
  \newtheorem{corollary}{Corollary}

  \usepackage{authblk}

  \usepackage{hyperref}
\else
  \documentclass[twoside,legno,twocolumn]{article}
  \usepackage{ltexpprt}
  \newcommand{\url}[1]{{\tt #1}}
  \newcommand{\href}[2]{{\tt #2}}
\fi

\usepackage{graphicx}\graphicspath{{figures/}}
\usepackage{enumerate, cite}
\usepackage{amsmath, amssymb, color}
\usepackage{todonotes}

\newcommand{\csqrtn}{\left\lceil \sqrt{n}\hspace{0.2em} \right\rceil}
\long\def\omitsection#1{\relax}

\newcommand{\abs}[1]{{\left\vert#1\right\vert}}
\DeclareMathOperator{\plot}{plot}
\DeclareMathOperator{\board}{board}

\ifFull
  \title{Small Superpatterns for Dominance Drawing\thanks{This research was supported in part by the National Science Foundation under grants 0830403 and 1217322.}}
  \author{Michael J. Bannister}
  \author{William E. Devanny}
  \author{David Eppstein}
  \affil{Department of Computer Science, University of California, Irvine}
\else
  \title{\Large Small Superpatterns for Dominance Drawing\thanks{This research was supported in part by the National Science Foundation under grants 0830403 and 1217322.}}
  \author{Michael J. Bannister \and
          William E. Devanny \and
          David Eppstein}
\fi
  
\begin{document}

\date{}
\maketitle\thispagestyle{empty}\setcounter{page}{0}

\begin{abstract}
\ifFull\else
\small\baselineskip=9pt
\fi
We exploit the connection between dominance drawings of directed acyclic graphs and permutations, in both directions, to provide improved bounds on the size of universal point sets for certain types of dominance drawing and on superpatterns for certain natural classes of permutations. In particular we show that there exist universal point sets for dominance drawings of the Hasse diagrams of width-two partial orders of size $O(n^{3/2})$, universal point sets for dominance drawings of $st$-outerplanar graphs of size $O(n\log n)$, and universal point sets for dominance drawings of directed trees of size $O(n^2)$. We show that $321$-avoiding permutations have superpatterns of size $O(n^{3/2})$, riffle permutations ($321$-, $2143$-, and $2413$-avoiding permutations) have superpatterns of size $O(n)$, and the concatenations of sequences of riffles and their inverses have superpatterns of size $O(n\log n)$. Our analysis includes a calculation of the leading constants in these bounds. 
\end{abstract}

\ifFull
\pagestyle{plain}
\newpage
\fi

\section{Introduction}
The \emph{universal point set problem} asks for a sequence of point sets $U_n$ in the plane such that every $n$-vertex planar graph can be straight-line embedded with vertices in $U_n$ and such that the cardinality of $U_n$ is as small as possible. Known upper bounds on the size of $U_n$ are quadratic, with a constant that has been improved from $1$ to $1/4$ over the last 25 years~\cite{DeFPacPol-STOC-88,Sch-SODA-90,Bra-TGGT-08,BanCheDev-GD-2013}. Surprisingly, the best lower bounds on the size of $U_n$ are only linear~\cite{ChrKar-SN-89,Kur-IPL-04,Mon-EaPGoaGP-12}. Reconciling this gap is a fundamental open problem in graph drawing~\cite{BraEppGoo-GD-03,DemORo-TOPP-12,Moh-OPG-07}.

Although the standard universal point set problem asks for point sets supporting planar straight line drawings of planar graphs, many researchers have previously asked similar questions about other drawing styles and other classes of graphs. Dujmovi\'{c} \emph{et al.} considered universal point sets for planar graphs with bends in the edges, where bend points must be placed on points in the universal set, showing that there exist point sets of size $O(n)$ for three bends, $O(n\log n)$ for two bends, and $O(n^2 / \log n)$ for one bend~\cite{DujEvaLaz-CG-2013}. If bends may be placed freely in the plane, then a construction of Everett \emph{et al.} using $n$ points is universal for planar graphs~\cite{EveLazLio-DCG-10}. Angelini \emph{et al.} constructed universal point sets of size $n$, lying on a parabolic path, for planar graphs where the edges are drawn as circular arcs~\cite{AngEppFra-CCCG-2013}. If edges may consist of two smoothly connected semicircles, then $n$ collinear points are universal for planar graphs~\cite{GioLioMch-ISAAC-07,BekKauKob-GD-12}. As well as for planar graphs, universal point sets have been considered for outerplanar graphs~\cite{GriMohPac-AMM-91}, planar partial 3-trees~\cite{FulTot-WADS-13}, simply-nested planar graphs~\cite{AngDibKau-GD-2012,BanCheDev-GD-2013}, pseudoline arrangement graphs~\cite{Epp-GD-13}, and planar graphs of bounded pathwidth~\cite{BanCheDev-GD-2013}.

\ifFull
\begin{figure}[b]
\else
\begin{figure*}
\fi
\centering\includegraphics[width=0.95\textwidth]{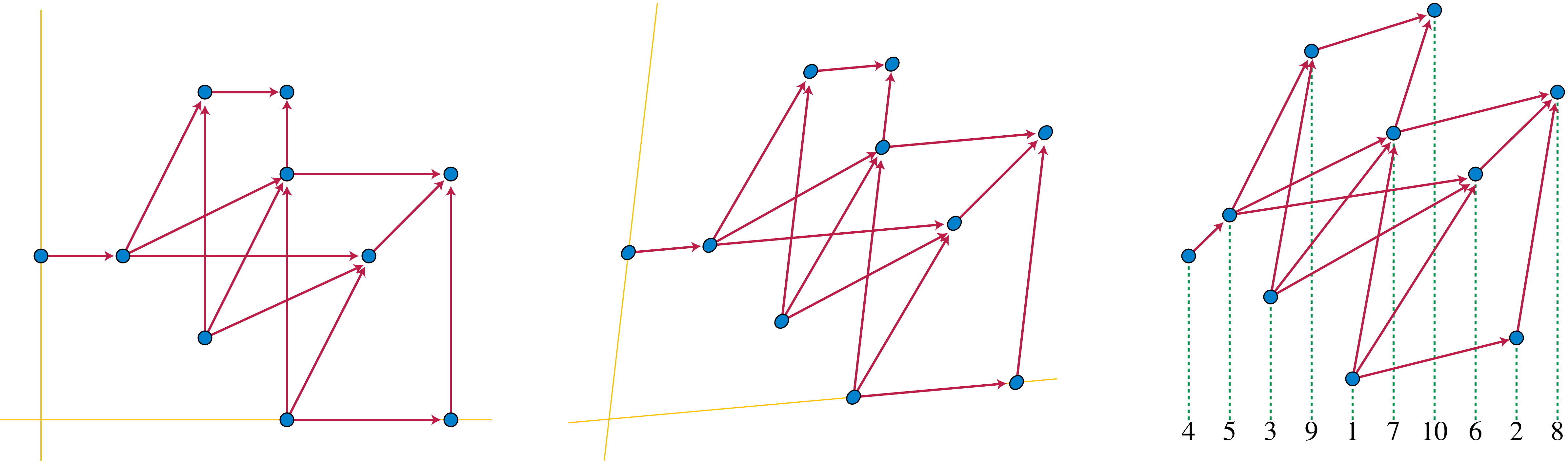}
\caption{Conversion of a dominance drawing (left) to a permutation (right) by performing a shearing transformation to eliminate pairs of points with equal coordinates, and then replacing the point coordinates by their positions in the sorted ordering of the coordinate values}
\label{fig:dom2perm}
\ifFull
\end{figure}
\else
\end{figure*}
\fi

In this paper we construct universal point sets for \emph{dominance drawing}, a standard style of graph drawing for directed acyclic graphs in which
\emph{reachability} (whether there is a path from $u$ to $v$)
must be the same relation as \emph{dominance} (whether both coordinates of $v$ are greater than or equal to both coordinates of~$u$). Equivalently, each edge must be oriented upwards, rightwards, or both, and every two vertices $u$ and~$v$ that form the lower left and upper right corners of an empty axis-aligned rectangle must be adjacent. A planar directed acyclic graph with a single source~$s$ and sink~$t$ has a dominance drawing if and only if it is a \emph{planar $st$-graph}, a graph in which $s$ and $t$ share a face in at least one embedding of the graph; if in addition the graph is transitively reduced, then its drawing is automatically non-crossing~\cite{GD-Book,DiBTamTol-DCG-92}.

Any dominance drawing can be transformed by shearing into a dominance drawing in which no two points share an $x$ or $y$ coordinate. Changing the coordinates of the points without changing their coordinate-wise sorted ordering preserves the properties of a dominance drawing, so we may assume without loss of generality that the coordinates $(x_i,y_i)$ of the points are distinct integers from $1$ to $n$.
In this way, every dominance drawing corresponds to a permutation $\sigma$ of size $n = \abs{V}$ that maps $x_i$ to $y_i$ (Figure~\ref{fig:dom2perm}). This connection between dominance drawing and permutations was exploited by Bannister \emph{et al.} to construct universal point sets of size $n^2 / 2 + \Theta(n)$ for dominance drawings of planar $st$-graphs~\cite{BanCheDev-GD-2013}. This construction relies on results in the study of permutation patterns. A permutation $\pi$ is said to be a \emph{pattern} of a permutation $\sigma$, if there exists a subsequence of $\sigma$ with elements having the same relative ordering as the elements in $\pi$. Mathematicians working in permutation patterns have defined \emph{superpatterns} to be permutations which contain all length $n$ permutations as patterns, and have studied bounds on the size of such permutations~\cite{Arr-EJC-99,EriEriLin-AoC-07,Mil-JCTSA-2009}. The current best upper bound for the size of a superpattern is $n^2 / 2 + \Theta(n)$~\cite{Mil-JCTSA-2009}. The set of points $(i,\sigma_i)$ derived from a superpattern~$\sigma$ forms a universal point set of size $n^2 / 2 + \Theta(n)$ for dominance drawings of transitively reduced planar $st$-graphs, and more generally for all graphs that have dominance drawings.~\cite{BanCheDev-GD-2013}.

In their investigation of universal point sets, Bannister \emph{et al.}~\cite{BanCheDev-GD-2013} generalized superpatterns to $P$-super\-patterns, permutations containing as a pattern every permutation in a set~$P$.  Given a set $F$ of ``forbidden patterns'' let $S_n(F)$ be the set of all length-$n$ permutations avoiding all patterns in $F$. Bannister \emph{et al.{}} constructed $S_n(213)$-superpatterns of size $n^2 / 4 + \Theta(n)$, and used them to produce universal point sets for planar straight-line drawings of size $n^2 / 4 - \Theta(n)$. They also showed that every proper subclass of the $213$-avoiding permutations has near-linear superpatterns and used them to construct near-linear universal point sets for planar graphs of bounded pathwidth. However, as they observed, the permutations arising from dominance drawings of planar $st$-graphs have no forbidden patterns, preventing the smaller superpatterns arising from forbidden patterns from being used in dominance drawings of these graphs.

\paragraph{New results.}
In this paper, we again consider dominance drawings, of both $st$-planar and non-$st$-planar graphs. We show that, unlike for arbitrary $st$-planar graphs, other important classes of graphs have forbidden patterns in the permutations defined by their dominance drawings. Based on this observation, we extend the connection between dominance drawing and superpatterns to construct smaller universal point sets for dominance drawings of these graph classes. Specifically, our contributions include the construction of

\begin{itemize}\itemsep0pt
\item universal point sets for planar dominance drawings of directed trees of size $n^{2}/4 + \Theta(n)$, based on the $213$-avoiding superpatterns from our previous work~\cite{BanCheDev-GD-2013};
\item superpatterns for $321$-avoiding permutations of size $22n^{3/2} + \Theta(n)$;
\item universal point sets for non-planar dominance drawings of the Hasse diagrams of width-$2$ partial orders of size $22n^{3/2} + \Theta(n)$, based on $321$-avoiding superpatterns;
\item superpatterns for the permutations corresponding to \emph{riffle shuffles} of size $2n-1$;
\item universal point sets for dominance drawings of $st$-outerplanar graphs of size $32n\log n+\Theta(n)$, using our superpatterns for riffles and their inverses; and
\item superpatterns for the concatenations of riffles and their inverses, of size $16n\log n+\Theta(n)$, using a simplified version of our universal point sets for $st$-outerplanar graphs.

\end{itemize}

Although of superlinear size, our universal point sets may all be placed into grids with low area. The $st$-outerplanar universal point sets lie in a $O(n) \times O(n\log n)$ grid, and all of our other universal points lie in an $O(n)\times O(n)$ grid, leading to compact drawings of the graphs embedded on them.


\paragraph{Application.}
The visualization of trees is so frequent as to need no additional motivation. Width-two Hasse diagrams and $st$-outerplanar graphs are less common, but both may arise in the visualization of change histories in distributed version control systems such as git or mercurial. If two editors of a  project repeatedly pull changes from or push them to a shared master repository, then the dependency graph of their version histories will form an $st$-outerplanar graph: their local changes form two paths and the pull events produce edges from one path to the other. If, on the other hand, two editors each maintain both a current version and an older stable version of the project, and each editor sometimes synchronizes his or her current version with the stable version of the other, then a more general width-2 Hasse diagram could result. Larger numbers of editors would lead to partial orders with higher width, beyond the scope of our study.

\section{Preliminaries and notation}

\paragraph{Permutations, patterns, and superpatterns.}

We denote the set of all permutations of the numbers from $1$ to $n$ by $S_n$, and we specify a given permutation as a string of numbers, e.g.,
\ifFull$\else\[\fi
S_3 = \{123, 132, 213, 231, 312, 321\}.
\ifFull$\else\]\fi
If $\sigma$ is a permutation then we will write $\sigma_i$ for the $i^\text{th}$ position of $\sigma$ (starting from position~$1$), and $\abs{\sigma}$ for the number of elements in $\sigma$. For example, if $\sigma = 2143$, then $\sigma_2 = 1$ and $\abs{\sigma} = 4$. We also define the plot of a permutation $\sigma$ as the set of $|\sigma|$ points in the plane given by $\plot(\sigma) = \{(i,\sigma_i) \mid 1 \leq i \leq \abs{\sigma}\}$.

Permutation $\pi$ is a pattern of permutation $\sigma$ if there exists a sequence of integers
\ifFull$\else\[\fi
1 \leq \ell_1 \leq \ell_2 \leq \cdots \leq \ell_\abs{\pi}
\ifFull$\else\]\fi
such that  $\pi_i < \pi_j$ if and only if $\sigma_{\ell_i} < \sigma_{\ell_j}$. Equivalently, $\pi$ is a pattern of $\sigma$ if $\pi$ has the same order type as a subsequence of $\sigma$. A permutation $\sigma$ \emph{avoids} a permutation $\pi$ if $\pi$ is not a pattern of $\sigma$. We denote by $S_n(\pi_1, \pi_2, \ldots, \pi_k)$ the set of length-$n$ permutations avoiding all the patterns $\pi_1, \pi_2, \ldots, \pi_k$.
A \emph{permutation class} is a set of permutation closed under taking patterns, meaning that every pattern of a permutation in the class is also in the class. Every permutation class may be defined by a (possibly infinite) set of \emph{forbidden patterns},  the minimal patterns not belonging to the class and  are, thus, avoided by every permutation in the class. Given a set $P$ of permutations we define a $P$-superpattern to be a permutation $\sigma$ such that every $\pi \in P$ is a pattern of $\sigma$. In particular, we will frequently consider $S_n(F)$-superpatterns for sets $F$ of forbidden patterns.

When working with permutation patterns it will be convenient to work with a compressed form of the plot of a permutation, its chessboard representation. Define the \emph{columns} of a permutation $\sigma$ to be the maximal ascending runs in $\sigma$ (contiguous subsequences that are in monotonically increasing numerical value). Define the \emph{rows} of a permutation $\sigma$ to be the maximal (non-contiguous) subsequences of $\sigma$ that form ascending runs in $\sigma^{-1}$. A row contains a maximal  set of contiguous numerical values that appear in sorted order in $\sigma$.
A \emph{block} of a permutation $\sigma$ is a contiguous subsequence of $\sigma$ containing consecutive values. For example, $543$ is a decreasing block in $\sigma = 6154372$. The intersection of a column and a row form a (possibly empty) increasing block. The \emph{chessboard representation} of a permutation $\sigma$ is a matrix $M$ where $M_{i,j}$ is the number of elements of $\sigma$ belonging to the intersection of the $i^\text{th}$ row and $j^\text{th}$ column. (This differs from a related definition of the chessboard representation by Miller~\cite{Mil-JCTSA-2009}, using descending rows rather than ascending rows, for which the intersection of a row and a column can contain at most one element.) As defined here, the chessboard representation respects the dominance relations in $\plot(\sigma)$. Indeed, the transformation from $\plot(\sigma)$ to $\board(\sigma)$ (shown in Figure~\ref{fig:perm-to-chessboard}) corresponds to a standard form of grid compaction used as a post-processing step in the construction of dominance drawings~\cite{GD-Book}.

\ifFull
\begin{figure}[h]
\else
\begin{figure*}
\fi
\centering
\includegraphics[width=0.75\textwidth]{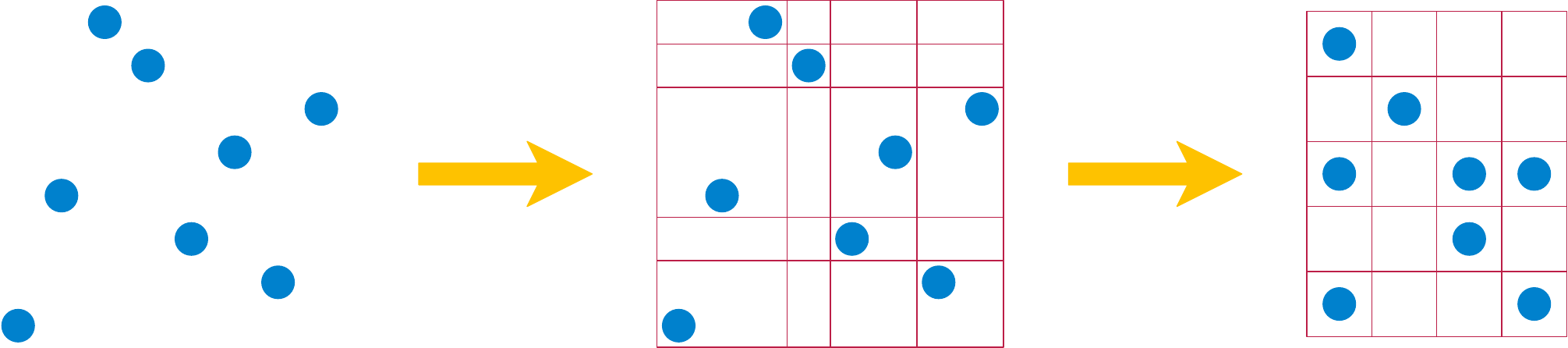}
\caption{The conversion from $\plot(\sigma)$ (left) to $\board(\sigma)$ (right, shown graphically by representing  0 as an empty square and 1 as a dot) via the partition of $\sigma$ into columns and rows (center) for $\sigma = 14 873526$.}
\label{fig:perm-to-chessboard}
\ifFull
\end{figure}
\else
\end{figure*}
\fi

\paragraph{Partial orders and order dimension.}
The study of non-planar dominance drawing leads us to the study of partial orders and their dimension. Given a set $A$, a \emph{partial order} on $A$ is a binary relation $\leq$ satisfying the following three properties:
\begin{itemize}\itemsep0pt
\item for all $a \in A$,  $a \leq a$ (reflexivity),
\item for all $a,b \in A$, $a \leq b$ and $b \leq a$ implies $a = b$ (antisymmetry) , and
\item for all $a,b,c \in A$, $a \leq b$ and $b \leq c$ implies $a \leq c$ (transitivity).
\end{itemize}
If, in addition to these properties either $a \leq b$ or $b \leq a$ for every $a,b \in A$, then the order is a \emph{total order} or \emph{linear order}. The \emph{intersection} of two partial orders on the same set $A$ is a binary relation consisting of the pairs that have the same relation in both given orders.
The \emph{width} of a partial order is the maximum cardinality of an \emph{antichain}, a set in which no two elements are comparable (related to each other by $\leq$); by Dilworth's theorem, the width equals the smallest number of \emph{chains} (totally ordered subsets) into which the elements can be partitioned~\cite{Dil-AM-50}.
The \emph{dimension} of a partial order $P$ is the least integer $k$ such that $P$ can be described as the intersection of $k$ total orders~\cite{DusMil-AJM-41}; it is never greater than $P$'s width~\cite{Hir-SRK-1951}.

Every partial order can be described uniquely by a transitively closed directed acyclic graph in which the vertices correspond to the elements of the partial order,
and there is an edge from $u$ to $v$ if and only if $u\leq v$ and $u\ne v$.
Alternatively, every partial order can be described uniquely by a transitively reduced directed acyclic graph,
its \emph{covering graph}, on the same vertex set.
In the covering graph, there is an edge from $u$ to $v$ if and only if $u\leq v$, $u\ne v$, and there does not exist $w$ with $u\leq w\leq v$, $u\ne w$, and $v\ne w$.
The partial order itself can be recovered as the reachability relation in either of these two graphs.
The \emph{Hasse diagram} of a partial order is a drawing of its covering graph. A transitively reduced directed acyclic graph has a (possibly nonplanar) dominance drawing if and only if its induced partial order has order dimension two: such a drawing may be obtained by finding two total orders whose intersection is the given partial order, using positions in one of these two total orders as the $x$-coordinates of the points, and using positions in the other total order as the $y$-coordinates of the points. Conversely, in a dominance drawing of a Hasse diagram, the sorted orderings of the points by their coordinates give two total orders whose intersection is the depicted partial order. In particular, because the partial orders of width two also have dimension at most two, their Hasse diagrams always have dominance drawings. The graphs represented by these drawings are planar, but in general not $st$-planar, and the drawings may have crossings.

\paragraph{Planar and outerplanar DAGs.}
A graph is a \emph{planar $st$-graph}, or more concisely \emph{$st$-planar}, if it is a directed acyclic graph that has a single source $s$ (a vertex with no incoming edges) and sink $t$ (a vertex with no outgoing edges), both of which belong to the outer face of some planar embedding of the graph. A planar graph with a single source and sink has a dominance drawing if and only if it is $st$-planar~\cite{GD-Book}. In an $st$-planar graph, the partial ordering of the vertices by reachability forms a \emph{lattice}: each pair of vertices has a unique join, the closest vertex that they can both reach, and a unique meet, the closest vertex that can reach both of them. More strongly, a transitively reduced graph is $st$-planar if and only if it is the Hasse diagram of a two-dimensional lattice~\cite{Pla-JCTB-1976}. Every two-dimensional partial order can be extended to a lattice, a fact that Eppstein and Simons used to find confluent drawings of Hasse diagrams~\cite{EppSim-GD-11}, and that implies that the permutations corresponding to dominance drawings of $st$-planar graphs have no forbidden patterns~\cite{BanCheDev-GD-2013}.

However, subclasses of the $st$-planar graphs may nevertheless have forbidden patterns. One in particular that we consider here is the class of $st$-outerplanar graphs. First considered by Chlebus \emph{et al.}~\cite{ChlChrDik-FCT-87}, these are outerplanar DAGs with a single source and sink. By outerplanarity, every vertex belongs to the outer face, so in particular these graphs are $st$-planar.

\section{Superpatterns}

\paragraph{Riffles.}

The \emph{riffle shuffle permutations} are the permutation that can be created by a single riffle shuffle of a deck of $n$ playing cards. In terms of permutation patterns, they may be described as the permutations that avoid $321$, $2143$ and $2413$\cite{Atk-DM-1999}. We also define the \emph{antiriffle} shuffle permutations to be the permutations whose inverse is a riffle shuffle permutation. Since pattern containment is preserved under taking inverses the inverse of a riffle shuffle superpattern is an antiriffle superpattern.

\begin{theorem}
The riffle (antirffle) shuffle permutations have superpatterns of size $2n-1$.
\end{theorem}

\begin{proof}
Consider the permutation
\if$\else\[\fi
\rho_n = (n+1)1(n+2)2(n+3)3 \cdots (2n)n
\if$\else\]\fi
, constructed by performing a perfect riffle shuffle on a deck of $2n$ cards. We claim that every riffle permutation of length $n$ is a pattern of $\rho_n$. To see this, let $\sigma$ be an arbitrary riffle shuffle permutation of length $n$. Then by definition, $\sigma$ is formed by interleaving two sets $123\cdots (k-1)$ and $k\cdots n$ (where the first set may be empty). Call the first set the lower set and the second set the upper set. To embed $\sigma$ into $\rho_n$ we map $\sigma_i$ to $(\rho_n)_{2i}$ if $\sigma_i$ is in the upper set and $(\rho_n)_{2i-1}$ if $\sigma_i$ is in the lower set -- see Figure~\ref{fig:riffle-superpattern-example}.

The length of the superpattern may be reduced from $2n$ to $2n-1$ by the observation that the final element in $\rho_n$ does not need to be used, as the lower set can share a column with the last point in the upper set.
\end{proof}

\ifFull
\begin{figure}[h]
\else
\begin{figure*}
\fi
\centering
\includegraphics[height=0.75in]{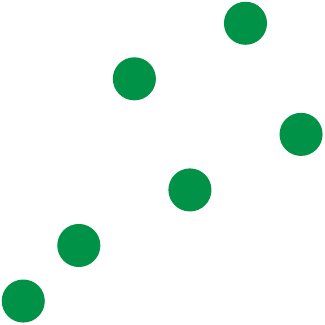}\hspace{4em}
\includegraphics[height=1.2in]{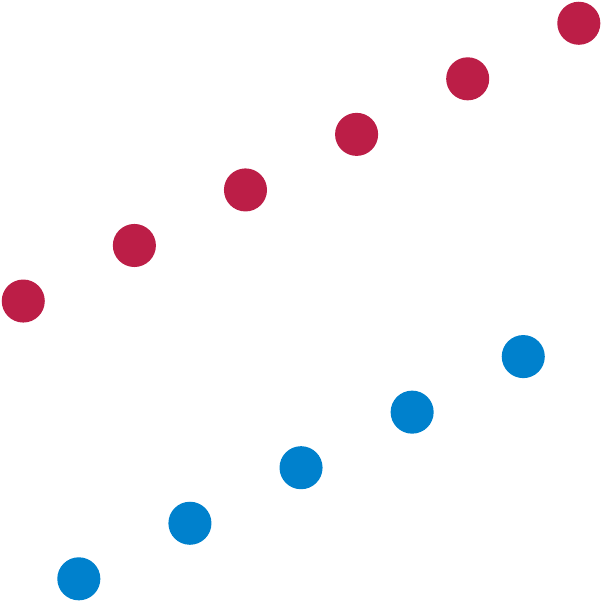}\hspace{4em}
\includegraphics[height=1.2in]{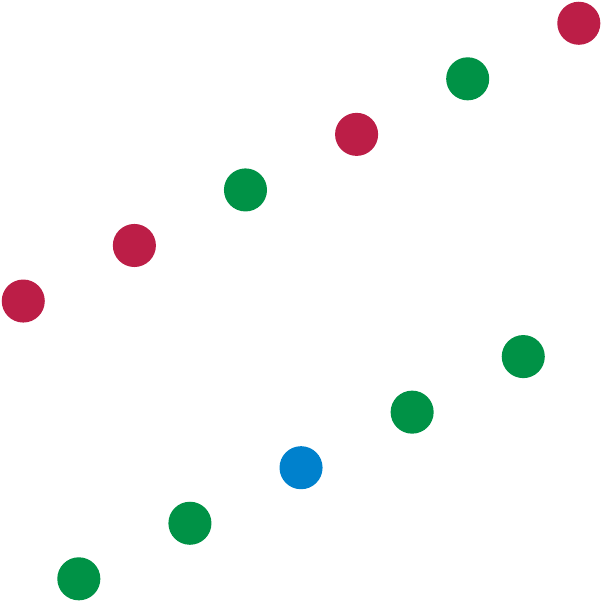}
\caption{Example of a length six riffle shuffle permutation (left), the superpattern $\rho_6$ (center), and an embedding of the permutation into the superpattern (right). }
\label{fig:riffle-superpattern-example}
\ifFull
\end{figure}
\else
\end{figure*}
\fi

\paragraph{321-avoiding permutations.} The $321$-avoiding permutations are precisely the \emph{$2$-increasing} permutations, meaning that we can partition them into two increasing subsequences, which we may choose so that one, the \emph{upper subsequence}, is above and to the left of the other, the \emph{lower subsequence}, in the plot of the permutation, and so that the upper sequence is maximal with this property.
In this section we construct $S_n(321)$-superpatterns of size $O(n^{3/2})$.   We begin our construction with an embedding into a $2n$ by $2n$ grid.

\begin{lemma}\label{lem:start-embed}
Let $\sigma \in S_n(321)$. Then $\sigma$ can be embedded in the bottom right triangle of a $2n$ by $2n$ grid with the upper subsequence of $\sigma$ on the diagonal.
\end{lemma}
\begin{proof}
Start with $\plot(\sigma)$ in an $n$ by $n$ grid.  Consider the points of the upper subsequence in left to right order.  For each of these points that is below the diagonal, shift it and all the points above it upward by the same amount, so that its shifted location lies on the diagonal. Similarly, for each point of the upper subsequence that is above the diagonal shift it and all the points to the right of it rightward, so that its shifted position lies on the diagonal.  After this shifting the spacing between any two successive points on the diagonal is equal to the $L^\infty$-distance between the two points in $\plot(\sigma)$.  The sum of these distances in $\plot(\sigma)$, including implicit starting and ending points at $(0,0)$ and $(n,n)$, cannot be more than $2n$.  So, after the shifting process is complete, the points are contained in a $2n$ by $2n$ grid.
\end{proof}

We define our superpattern $\mu_n$ by specifying its $6n + 8\csqrtn$ by $6n + 8\csqrtn$ chessboard representation $M_n$. The entries of $M_n$ are all zero except:
\begin{enumerate}
\item $[M_n]_{i,j} = 1$ if $i - j \leq 2\csqrtn + 1$ for $1 \leq i,j \leq 6n  + 8\csqrtn$ (red band);
\item $[M_n]_{i,j} = 1$ if $j - i = k \csqrtn$ for $0 \leq k \leq \left\lfloor \frac{6n  + 8\csqrtn}{\csqrtn}\right\rfloor$ for $1 \leq i,j \leq 6n + 8\csqrtn$ (blue lines).
\end{enumerate}
See Figure~\ref{fig:321-superpat-ex} for the high level structure. The red band in $M_n$ is made up of those entries below the diagonal whose $L^\infty$-distance to the diagonal is at most $2\csqrtn + 1$, and the blue lines are made up of those entries above the diagonal whose $L^\infty$-distance to the diagonal is a multiple of $\csqrtn$ and the diagonal itself.


\begin{lemma}
\label{lem:321-length}
The permutation $\mu_n$ is of length $30n^{3/2} + \Theta(n)$.
\end{lemma}
\begin{proof}
The length of $\mu_n$ is equal to the number of nonzero entries in $M_n$. The number of entries in the red band and blue lines are given respectively by the two arithmetic series
\[
\sum_{i=1}^{2\csqrtn + 1} 6n  + 8\csqrtn - i = 12n^{3/2} +\Theta(n)
\]
and
\[
\sum_{i = 0}^{\left\lfloor \frac{6n  + 8\csqrtn}{\csqrtn}\right\rfloor} 6n + 8\csqrtn - \csqrtn i = 18n^{3/2} + \Theta(n).
\]
Therefore, the length of $\mu_n$ is $30n^{3/2} + \Theta(n)$.
\end{proof}

\begin{figure}
\centering
\includegraphics{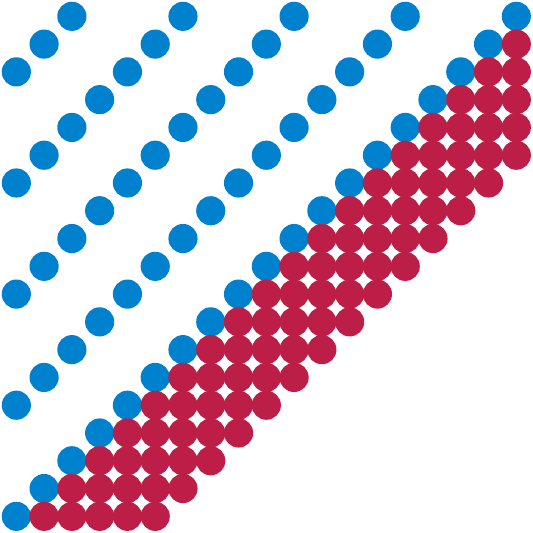}
\caption{Structure of the 321-avoiding superpattern}
\label{fig:321-superpat-ex}
\end{figure}

\begin{lemma}\label{lem:sum-analysis}
Let $d_i$ be a sequence such that $\sum_i d_i \leq cn$ with $d_i \geq \frac{\csqrtn}{2}$ for all i.  Then $\sum_i \left\lceil \frac{d_i}{\csqrtn} \right \rceil \leq 2c\csqrtn.$
\end{lemma}
\begin{proof}
Let $S$ be the set $\left\{ i \mid d_i \leq \csqrtn\right\}$ and $L$ be the set $\left\{ i \mid d_i > \csqrtn \right\}$. Then, partitioning the sum $\sum_i d_i$ into two subsequences and then applying the assumption of the lemma that each term is sufficiently large,
\[
c n\geq \sum_i d_i  = \sum_{i\in S} d_i + \sum_{i \in L} d_i
   \geq \frac{|S|\csqrtn}{2} + \sum_{i \in L} d_i,\]
or equivalently,
\ifFull$\else\[\fi
\sum_{i \in L} d_i \leq cn - \frac{|S|\csqrtn}{2}.
\ifFull$\else\]\fi
Plugging this bound into the sum from the conclusion of the lemma, and using the facts that
each term of this sum for an element of $S$ rounds up to one and that each term for an element of $L$ is rounded up by at most one unit, we get
\ifFull
\[
\sum_i \left\lceil \frac{d_i}{\csqrtn} \right \rceil \leq |S| + |L| + \frac{\sum_{i\in L} d_i}{\csqrtn}
   \leq |S| + |L| + \frac{cn - \frac{|S|\csqrtn}{2}}{\csqrtn}
   \leq \frac{|S|}{2} + |L| + c\csqrtn.
\]
\else
\begin{align*}
\sum_i \left\lceil \frac{d_i}{\csqrtn} \right \rceil
&\leq |S| + |L| + \frac{\sum_{i\in L} d_i}{\csqrtn}\\
&\leq |S| + |L| + \frac{cn - \frac{|S|\csqrtn}{2}}{\csqrtn}\\
&\leq \frac{|S|}{2} + |L| + c\csqrtn.
\end{align*}
\fi
Now if $|L| = c\csqrtn - k$ for some $k$, then $|S| \leq 2k$, for otherwise (again using the lower bound on the values in $S$ in the statement of the lemma) $\sum d_i$ would be larger than $cn$.  Therefore, we can simplify the right hand side of the above inequality, giving
$$\sum_i \left\lceil \frac{d_i}{\csqrtn} \right \rceil  \leq 2c\csqrtn$$
as required.
\end{proof}

\begin{theorem}\label{thm:321-superpat}
The permutation $\mu_n$ is a $S_n(321)$-superpattern.
\end{theorem}
\begin{proof}
Let $\sigma$ be an arbitrary permutation in $S_n(321)$; we will show that $\mu_n$ contains $\sigma$ as a pattern by embedding $\sigma$ in $M_n$ while preserving the relative ordering of its elements.

To embed $\sigma$, we define two shift operations that alter an embedding $\sigma$ in the grid while still respecting the order relations of its points.  The first is an \emph{upwards shift} at an element $x$.  This shift moves $x$ and every element above it upwards by $\csqrtn$ units.  Similarly a \emph{rightwards shift} at $x$ moves $x$ and every element to the right of it rightwards by $\csqrtn$ units.

Apply Lemma~\ref{lem:start-embed} to embed $\sigma$ into the bottom left $2n$ by $2n$ points of a $6n + 8\csqrtn$ by $6n + 8\csqrtn$ grid, with the upper subsequence of $\sigma$ on a blue line of $M_n$, the main diagonal.  Next,  consider each point of the lower subsequence in left to right order. If one of these points is outside the red band, perform shift operations to move it into the red band.  If the point is below the red band by $i$ units,  perform $\left\lceil i/\csqrtn + 1/2 \right\rceil$ upwards shifts.  If the point is above the red band by $i$ units,  perform $\left\lceil i/\csqrtn + 1/2 \right\rceil$ rightwards shifts.  After these moves, the point in consideration is in the middle half of the red band and the relative ordering of the points remains the same. No element can be shifted rightwards more times than it was shifted upwards, so the elements of the upper subsequence will always remain on the blue lines during the shift operations.  After these shifts, the entire lower subsequence will lie within the red band. 

To complete the proof we need only show that after the specified moves all points remain in a $6n + 8\csqrtn$ by $6n + 8\csqrtn$ grid. Let $s_i$ be the sequence of indices of elements at which a shift takes place. Let $d_i$ be the $L_\infty$-distance between the Lemma~\ref{lem:start-embed} positions of $\sigma_{s_{i-1}}$ and $\sigma_{s_i}$. The sum of the $d_i$ is bounded above by $4n$, as Lemma~\ref{lem:start-embed} embeds the permutation in a $2n$ by $2n$ grid. Because a vertical shift at $s_i$ happens only when $d_i$ equals the horizontal displacement from $s_{i-1}$, and the sum of all of the horizontal displacements is $2n$, the sum of the $d_i$'s for the subsequence of elements causing vertical shifts is at most $2n$. Symmetrically, the sum of the $d_i$'s for the subsequence causing horizontal shifts is is at most $2n$. Also, each $d_i$ is at least $\csqrtn / 2$, as otherwise $\sigma_{s_i}$ would not have caused a shift. The number of shifts caused by $\sigma_{s_i}$ is at most $\lceil d_i / \csqrtn \rceil$.
If $H$ is the set of indices $i$ such that $\sigma_{s_i}$ causes a horizontal shift and $V$ is the set of indices $i$ such that $\sigma_{s_i}$ causes a vertical shift, then the total number of shifts can be bounded as
\ifFull
\[
\sum_{i \in H} \left\lceil \frac{d_i}{\csqrtn} \right\rceil  \leq 4\csqrtn\qquad\text{and}\qquad
\sum_{i \in V} \left\lceil \frac{d_i}{\csqrtn} \right\rceil  \leq 4\csqrtn,
\]
\else
\[
\sum_{i \in H} \left\lceil \frac{d_i}{\csqrtn} \right\rceil  \leq 4\csqrtn
\]
and
\[
\sum_{i \in V} \left\lceil \frac{d_i}{\csqrtn} \right\rceil  \leq 4\csqrtn,
\]
\fi
where the upper bounds follow from Lemma~\ref{lem:sum-analysis}.  Thus, elements are moved upwards and rightwards by at most $4\csqrtn \csqrtn \leq 4n + 8\csqrtn$ units, so all elements are within the $6n + 8\csqrtn$ by $6n + 8\csqrtn$ grid and $\sigma$ can be embedded into $M_n$.
\end{proof}

A close examination of Theorem~\ref{thm:321-superpat} will reveal that not every spot on the blue lines is used.  No point will ever be shifted upwards more than $2\csqrtn$ times, because no point starts further than $2n$ below the diagonal.  Therefore we only need the lowest $2\csqrtn$ blue lines.  Furthermore if in a particular column the bottom of the red band is $k > 0$ units away from the bottom of the grid, then no point in that column can be moved upwards by more than $k + \csqrtn/2$ units. The remaining usable portion of the blue lines is depicted in Figure~\ref{fig:constant-shave}.

\begin{theorem}\label{thm:321-constant}
There is a $S_n(321)$-superpattern of size $22n^{3/2} + \Theta(n)$.
\end{theorem}
\begin{proof}
If the points on the blue lines that cannot be used are removed from the superpattern, the total number of points on the blue lines that remain becomes less than or equal to
\ifFull
\[
\sum_{i = 0}^{2\csqrtn} 6n + (8 - i)\csqrtn = 10n^{3/2} + O(n),
\]
\else
\begin{multline*}
\sum_{i = 0}^{2\csqrtn} 6n + (8 - i)\csqrtn = 10n^{3/2} + O(n),
\end{multline*}
\fi
from which it follows that  the total size of the reduced superpattern is $22n^{3/2} + O(n)$.
\end{proof}

\begin{figure}[t]
\centering
\includegraphics[width=3in]{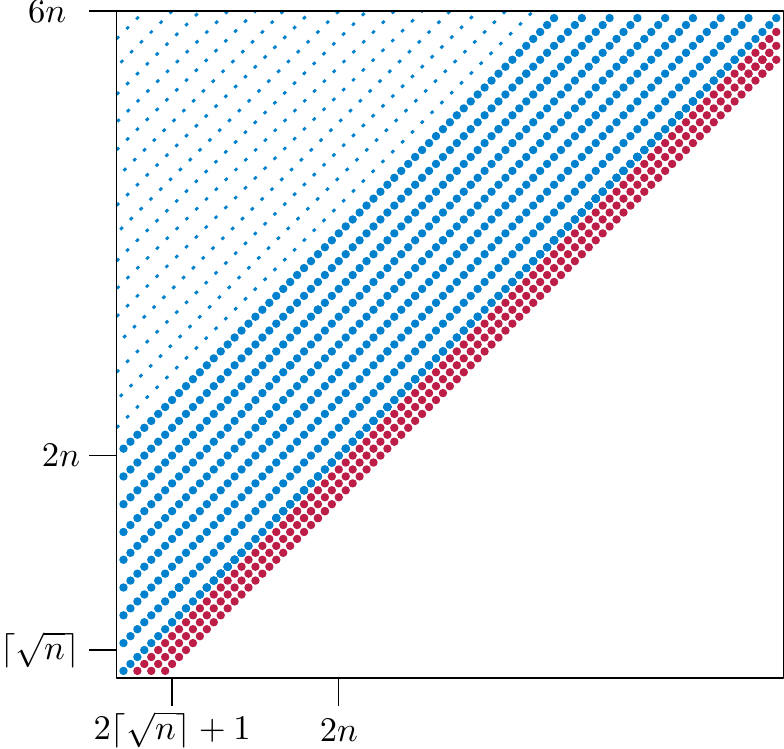}\qquad
\caption{The revised superpattern with removed pieces indicated by dotted lines.  Spacing and endpoints are only approximate.}
\label{fig:constant-shave}
\end{figure}

\section{Dominance drawing}
\subsection{$st$-outerplanar graphs}
Dominance drawings of $st$-outerplanar graphs induce permutations that avoid $321$, so by Theorem~\ref{thm:321-superpat} they have universal point sets of size $22n^{3/2} + \Theta(n)$. In this section we will improve on this construction and construct universal point sets of size $32n\log n + \Theta(n)$ for these drawings.

An $st$-outerplanar graph consists of two directed paths from $s$ to $t$ bordering the outer face, and edges  in the interior. Without loss of generality we may assume that the graph is transitively reduced, and the interior edges all connect one path to the other.
 If we draw the directed paths as parallel horizontal lines directed left to right, then we can classify the interior edges as upward or downward. In Figure~\ref{fig:st-out-struct} the contiguous regions of upward edges are highlighted in red, and the contiguous regions of downward edges are highlighted in blue. In addition to the blue and red regions we have green regions connecting the bottom row of a red region to the bottom row of the next blue region and yellow regions connecting the top row of a blue region to the top row of the next red region. The ``source paths'' into a red or blue region may be of any length including length zero (in which case the red and blue regions share a vertex), as shown at positions $B$ and $D$. The green and yellow paths need not be present, as shown in regions $C$ and $D$. However, if a green or yellow path is present it must have at least two edges, as otherwise the graph fails to be transitively reduced.

\paragraph{Computing the decomposition.} 
To compute this decomposition into colored regions we start with an outerplanar embedding of the given $st$-outerplanar graph. The outer face is bordered by two paths from $s$ to $t$, which might not be disjoint from each other. Among all possible embeddings, we choose one in which the first edge from one path to the other goes from the lower path to the higher path, and so that the orientation of the edges from one path to the other changes at each vertex shared by both paths. As in the figure, we use these two paths to partition the vertices into upper and lower \emph{rows}; vertices that are shared by both paths are placed only in one of the two rows.  A shared vertex with one incoming edge is placed on the same row as its predecessor; a shared vertex with two incoming edges (necessarily one from each row) is placed in such a way that the incoming edge from a different row has the same orientation as the preceding path-to-path edges. The vertices of a single row form a sequence of one or more contiguous paths; e.g. in the figure, the bottom row has two paths separated by a gap at $C$, while the top row has a gap at $D$.

\ifFull
\begin{figure}[t]
\else
\begin{figure*}
\fi
\centering
\includegraphics[width=5in]{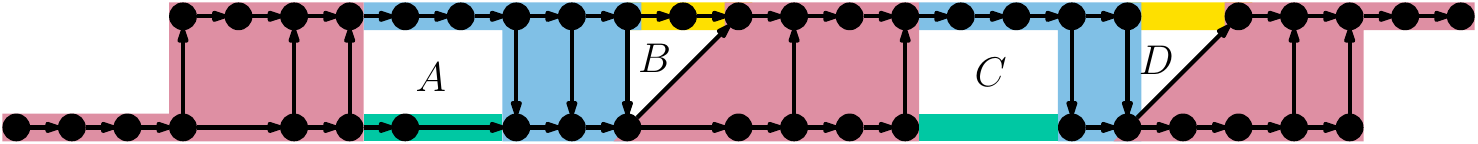}
\caption{Structure of an $st$-outerplanar graph.}
\label{fig:st-out-struct}
\ifFull
\end{figure}
\else
\end{figure*}
\fi

Next, we partition the graph into colored regions, as shown in the figure. To do so we iterate through the bottom row until we find a down arrow, a vertex of indegree two. The path between the last up arrow $e$ and this down arrow $f$ is classified as a green region. The vertices in the top row up to and including the destination of $e$ together with the vertices in the bottom row up to and including the source of $e$ form a red region. Once colored, these vertices are removed and we repeat the process this time iterating in a symmetric way through the top row to form a blue and yellow region. We continue in this way until all vertices have been colored.

\paragraph{Reachability properties of the decomposition.}
Consider the left-to-right sequences $R_i$ of red regions, $B_i$ of blue regions, $G_i$ of green regions, and $Y_i$ yellow regions. The union $R_i \cup B_i \cup G_i \cup Y_i$ will be referred to as the $i^\text{th}$ block. The following reachability facts then hold:
\begin{enumerate}
\item $R_i$ can reach all points in $B_i$ and $Y_i$, the bottom row of $R_i$ can reach all of $G_i$, and the top row can reach none of $G_i$;
\item $G_i$ can reach all points in the bottom row of $B_i$ and no points in $Y_i$, $R_i$ or the top row of $B_i$;
\item $B_i$ can reach no points in $R_i$ or $G_i$, the top row of $B_i$ can reach all points of $Y_i$, and the bottom row of $B_i$ can reach none of $Y_i$;
\item $Y_i$ can reach no points in $R_i$, $B_i$ or $G_i$;
\item $B_i$ can reach all points in $R_{i+1}$;
\item $Y_i$ can reach all points in the top row of $R_{i+1}$ and no points in $G_{i+1}$ or the bottom row of $R_i$.
\end{enumerate}
Facts 1--4 characterize the reachability within the $i^\text{th}$ block, and facts 5--6 characterize the reachability from the $i^\text{th}$ block to the $(i+1)^\text{th}$ block. Together they describe all the reachability relations in the graph.

A natural choice when producing a dominance drawing of the red and blue regions of a $st$-outerplanar graph is to keep the orientation of the upward edges and orient the downward edges rightward. In such a dominance drawing the red regions correspond to riffles and the blue regions correspond to antiriffles. This leads us to consider superpatterns for classes of permutations containing both riffles and antiriffles. Their individual superpatterns may be combined as shown in Figure~\ref{fig:riffle_suppat_cb}.

\begin{lemma}\label{lem:red-riffle}
A red region of size $n$ may be drawn on a riffle superpattern, and a blue region of size $n$ may be drawn on a antiriffle superpattern of size $n$.
\end{lemma}
\begin{proof}
To draw a red region we place all upward edges on points with the same $x$-coordinate while respecting the ordering induced by the top and bottom rows. The vertices not connected to an upward arrow are then placed such that between any two upward edges the vertices on the top row are placed with lower $x$-coordinates than the vertices on the bottom row. The embedding of a blue region into a antiriffle superpattern is similar, with downward edges being drawn left to right such that the source and destination have the same $y$-coordinate.
\end{proof}

Now we construct a universal point set $Q_m$ (for $m$ a power of two) for dominance drawings of $st$-outerplanar graphs; the notation hints at our quadtree-based construction of these sets. Starting with the riffle/antiriffle superpattern of side length $m = 2^k$ (a hollow square) we split the square into quarters, adding additional vertices to produce four overlapping riffle/antiriffle superpatterns of side length $2^{k-1}$. This part of the construction is represented by black points in Figure~\ref{fig:st-out-32}.

In addition, we add $k - 1$ compressed columns of yellow points, in the upper left riffle/antiriffle superpattern, of height $2^{i}$ at position $2^{k-1} - 2^{i} + 1$ for $1 \leq i \leq k-1$. The columns are vertically compressed such that they lie entirely between the first and third row. We also add green columns of symmetric size and position in the bottom right riffle/antiriffle superpattern. Finally, we continue recursively into the upper right riffle/antiriffle superpattern and the lower left riffle/antiriffle superpattern as illustrated in Figure~\ref{fig:st-out-32}.

\begin{figure}[t]
\centering
\includegraphics[height=1in]{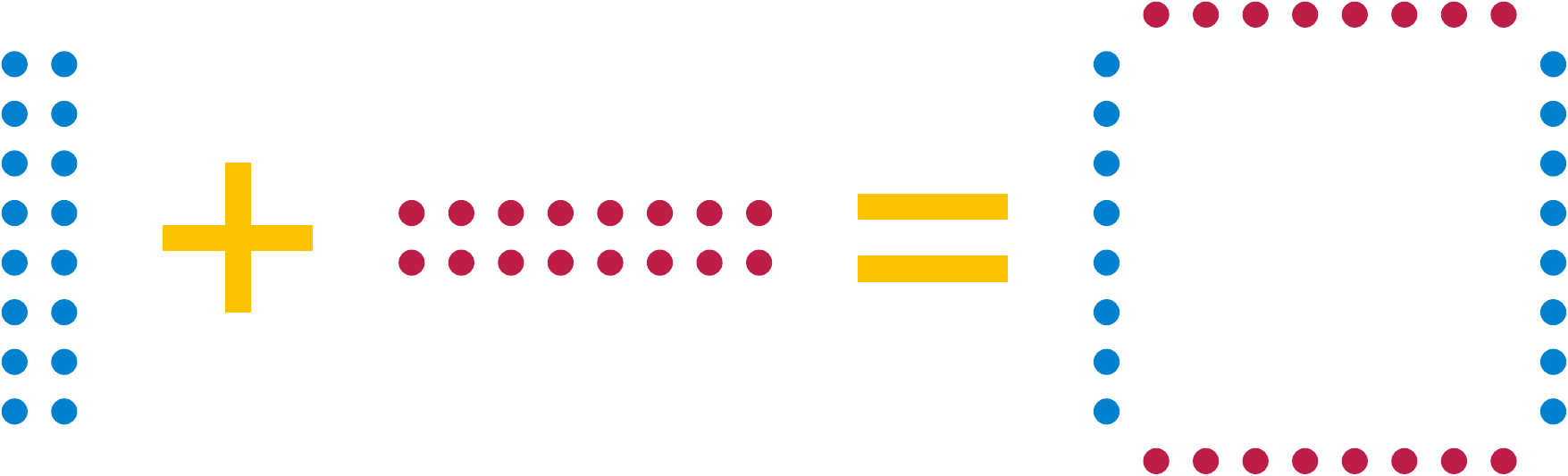}
\caption{Combining a riffle superpattern with an antiriffle superpattern, in their chessboard representations.}
\label{fig:riffle_suppat_cb}
\end{figure}

\begin{lemma}\label{lem:size-Qm}
The number of points in $Q_m$ is $4mk + 4m$ for $m = 2^k$.
\end{lemma}
\begin{proof}
The number of black and green/yellow points in $Q_m$ are given by
\ifFull
\[
4m + \sum_{i=0}^{k-2} 2^{i+1} \frac{m}{2^i} \leq 4m + 2mk
\qquad\text{and}\qquad
\sum_{i=0}^{k-2} 2(2^{i}-1)\frac{m}{2^{i}} \leq 2mk,
\]
\else
\[
4m + \sum_{i=0}^{k-2} 2^{i+1} \frac{m}{2^i} \leq 4m + 2mk
\]
and
\[
\sum_{i=0}^{k-2} 2(2^{i}-1)\frac{m}{2^{i}} \leq 2mk,
\]
\fi
respectively. Thus, the size of $Q_m$ is at most $4mk + 4m$.
\end{proof}

For $n$ a power of two define $\mathcal{I}_n$ to be a set of intervals, containing the interval $[1,n] \subseteq \mathbf{N}$ and closed under subdivision of an interval into two equal subintervals of size greater than one. For example,
\ifFull$\else\[\fi
\mathcal{I}_8 = \{[1,8], [1,4], [5,8], [1,2],[3,4],[4,5],[6,8]\}.
\ifFull$\else\]\fi
We can think of $\mathcal{I}_n$ as a being a one-dimensional quadtree.

\begin{lemma}\label{lem:quad-tree}
For any given finite sequence $a_i$ with $\sum_i a_i \leq n$, where $n$ is a power of two, there exists a sorted sequence of contiguous and disjoint intervals $I_i$ in $\mathcal{I}_{4n}$ with $\abs{I_i} \geq a_i$.
\end{lemma}
\begin{proof}
It suffices to show that there are disjoint intervals $I_i$ in $\mathcal{I}_{4n}$ with $\abs{I_i} \geq a_i$, as the intervals can be promoted (enlarged) until they become contiguous. Assume by induction on $k$ that we have found disjoint intervals $I_i$ for all $i<k$, satisfying the induction hypothesis that each of these intervals lies in the range $I_i \subseteq [1,\ell]$ where $\ell = \sum_{i<k} a_i$. To place $a_k$ first round it up to the nearest power of two; let $p$ be this rounded value. Now $\ell+2p\le \ell+4a_k$, and in the range $[\ell + 1, \ell + 2p]$ there is an interval $I_k \in \mathcal{I}_{4n}$ with $\abs{I_k} = p \geq a_k$. Thus, the lemma follows by induction.
\end{proof}

\ifFull
\begin{figure}[t]
\else
\begin{figure*}
\fi
\centering
\includegraphics[height=3in]{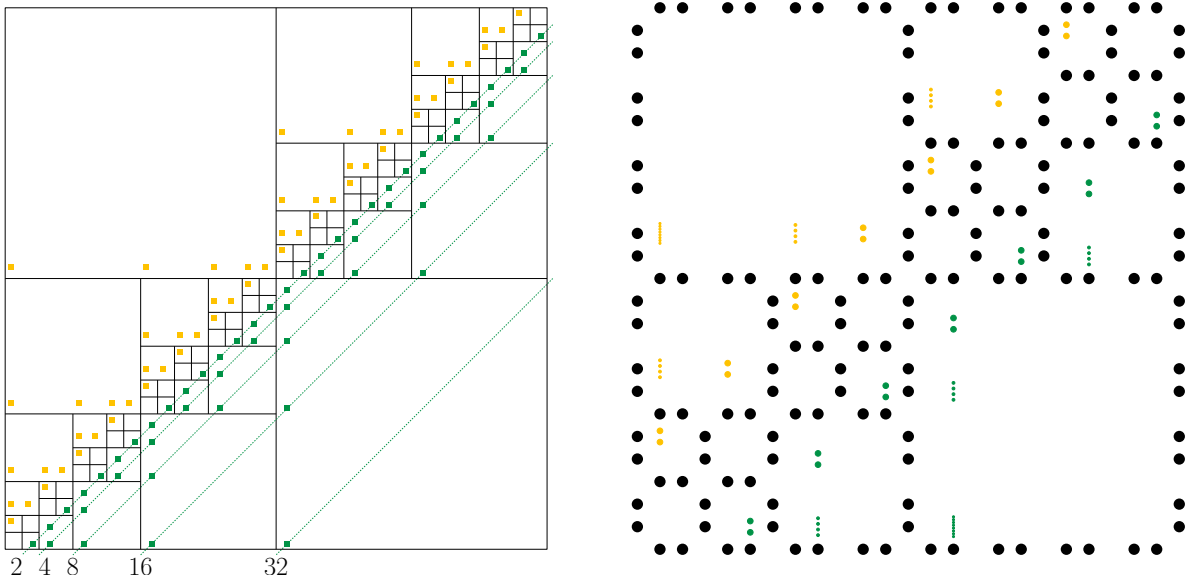}
\caption{High level view of a universal point set for $st$-outerplanar graphs with $m = 64$ (left), and a detailed view of a universal point set for $st$-outerplanar graphs with $m = 16$ (right). In the high level view the green and yellow boxes expand to columns of points. The sizes of the green boxes are given under their level curves.} 
\label{fig:st-out-32}
\ifFull
\end{figure}
\else
\end{figure*}
\fi

\begin{lemma}\label{lem:col-size}
If $S$ is a  square of side length $\ell$ on the main diagonal of $Q_m$, then there exists a column of $\ell$ yellow points two rows above $S$ and a column a green points two columns to the right of $S$.
\end{lemma}
\begin{proof}
By construction there are $\log \ell$ columns above $S$ the largest and furthest left having $\ell$ points. Similarly, there are $\log \ell$ columns to the right of $S$ the largest and lowest down having size $\ell$ points.
\end{proof}

\begin{lemma}\label{lem:dom-rel}
Let $S_1$ and $S_2$ be adjacent diagonal squares, and let $G_1, Y_1$ and $G_2, Y_2$ be the green and yellow columns of Lemma~\ref{lem:col-size} for $S_1$ and $S_2$. Then
\begin{enumerate}
\item the points in $G_1$ (respectively $Y_1$) are independent of the points in $Y_2$ (respectively $G_2$) with respect to the dominance relation;
\item the points in $G_1$ are dominated by the right side of $S_2$, but independent of the left side;
\item the points in $Y_1$ are dominated by the top side of $S_2$, but independent of the bottom side.
\end{enumerate}
\end{lemma}
\begin{proof}
Parts 2. and 3. follow directly from the construction. For part 1. first notice that the points of $G_1$ are placed to the left of $Y_2$. Now since $G_1$ is placed is placed two rows above $S_1$ and $Y_2$ is placed one row above $S_1$ we have that ever point in $G_1$ is above every point in $Y_2$. Thus, they are independent. The symmetric argument proves the independence of $Y_1$ and $G_2$.
\end{proof}

\begin{theorem}\label{thm:st-out-univ}
There exist universal point sets for dominance drawings of $st$-outerplanar graphs of size $32n\log n + \Theta(n)$. In particular, $Q_{8n}$ is a  universal point set for dominance drawings of $st$-outerplanar graphs.
\end{theorem}
\begin{proof}
We show that $Q_{4n}$ is universal for $st$-outerplanar graphs of size $n$, when $n$ is a power of two; it follows for this that $Q_{8n}$ is universal for all $n$. Given an $st$-outerplanar $G$ graph with $n$ vertices ($n$ a power of two) we partition it into blue $B_i$, red $R_i$, green $G_i$, and yellow $Y_i$ regions. We define the sequence $X_i$ such that $X_{2j + 1} = R_j \cup G_j$ and $X_{2j} = B_j \cup Y_j$ for $j \geq 1$. Since $\sum_i \abs{X_i} = n$, there exist by Lemma~\ref{lem:quad-tree} a set of disjoint contiguous intervals $I_i$ in $\mathcal{I}_{4n}$ with $\abs{I_i} \geq \abs{X_i}$. The intervals $I_i$ correspond to the recursively created riffle/antiriffle superpatterns in $Q_{4n}$ on the diagonal. These are points in $Q_{4n}$ we will use.

If $X_i$ is an upward region (red/green), then its red parts are drawn on the horizontal points of the riffle/antiriffle superpattern corresponding to $I_i$ and its green part is drawn on the green column of size $\abs{I_i}$ added to the right of the riffle/antiriffle superpattern. Symmetrically, if $X_i$ is a downward region (blue/yellow), then the blue parts are drawn on the vertical points of the riffle/antiriffle superpattern and its yellow part is drawn on the yellow column of size $\abs{I_i}$ added to above the riffle/antiriffle superpattern directly above. This yields a dominance drawing of $G$ in $Q_{4n}$, by Lemmas~\ref{lem:red-riffle}, \ref{lem:col-size} and \ref{lem:dom-rel}. It follows that $Q_{8n}$, a set of size $32n\log n + \Theta(n)$ by Lemma~\ref{lem:size-Qm}, is universal for dominance drawings of $st$-outerplanar graphs.
\end{proof}

The \emph{concatenation} or \emph{skew sum} $\sigma \ominus \tau$ of two permutations $\sigma $ and $\tau$ has $(\sigma \ominus \tau)_i$ equal to $\sigma_i$ for $1 \leq i \leq \abs{\sigma}$ and equal to $\tau_i + \abs{\sigma}$ otherwise. Define the set of $\emph{skew riffle}$ permutations to be the minimal set that contains the riffles and antiriffles and is closed under skew sums. Compared to the permutations for $st$-outerplanar graphs, the skew riffles have a simplified decomposition in which each element belongs uniquely to a riffle or antiriffle, allowing smaller superpatterns for these permutations.

\begin{theorem}
Skew riffles have superpattern of size $16n\log n + \Theta(n)$.
\end{theorem}

\begin{proof}
Removing the green and yellow segments from the universal point set produces the chessboard representation of a superpattern for skew riffles.
\end{proof}

\subsection{Partial orders of width two}
In this section we consider dominance drawings of partial orders of width two. 
Since such partial orders have dimension at most two, they have dominance drawings. An explicit dominance drawing may be obtained by partitioning the partial order into two chains (Dilworth's theorem), and placing each vertex $v$ at a point whose $x$ coordinate is the earliest position reachable from $v$ on one chain (or the position of $v$ itself, if it belongs to the chain) and whose $y$ coordinate is the earliest position reachable from $v$ on the other chain.

\begin{theorem}
Dominance drawings of $n$-element width two partial orders have universal sets of size $22n^{3/2} + \Theta(n)$.
\end{theorem}

\begin{proof}
The permutation corresponding to a dominance drawing of a width-two partial ordering must necessarily avoid the pattern $321$, as the elements of such a pattern would form an antichain of size three (Figure~\ref{fig:321-antichain}). Thus, if $\mu_n$ is an $S_n(321)$-superpattern, $\plot(\mu_n)$ is a universal point set for dominance drawings of $n$-element width two partial orders.
The result follows from Theorem~\ref{thm:321-constant}.
\end{proof}

\begin{figure}
\centering
\includegraphics[height=1.25in]{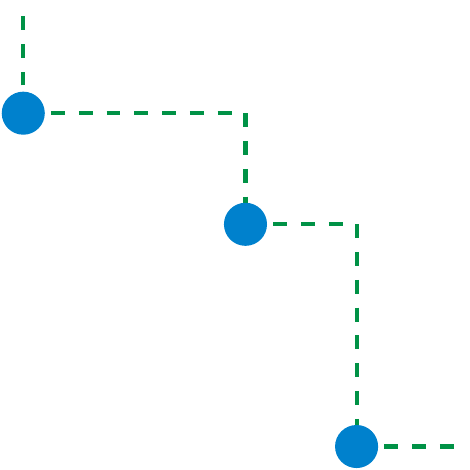}\qquad\qquad
\includegraphics[height=1.25in]{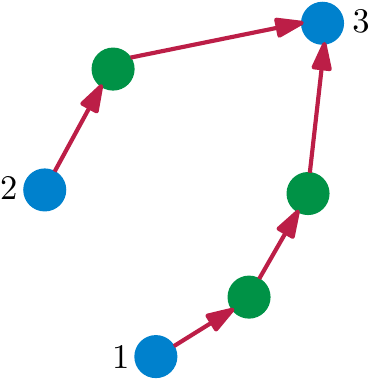} 
\caption{Left: If the dominance drawing of a partial order contains a 321 pattern then it has a antichain of length three. Right: A dominance drawing whose permutation has a 213 pattern cannot be a tree.}
\label{fig:321-antichain}
\label{fig:213-dom-tree}
\end{figure}

\subsection{Trees}
A \emph{directed tree} is a tree whose edges are directed from a root to its leaves. Such a graph can be transformed to a $st$-planar graph by adding a sink $t$ and directed edges from each leaf to $t$. Thus, every directed  tree has a planar dominance drawing. In addition, the permutation defined by a dominance drawing of a tree is $213$-avoiding: in a $213$ pattern, the vertex $v_3$ corresponding to the $3$ element dominates the other two vertices $v_1$ and $v_2$, so the corresponding graph would have paths from $v_1$ and $v_2$ to $v_3$, but no path from $v_1$ to $v_2$ nor vice versa; this is impossible in a directed tree (Figure~\ref{fig:213-dom-tree}).

\begin{theorem}
There exist universal point sets for planar dominance drawings of rooted trees of size $n^2/4+ \Theta(n)$.
\end{theorem}

\begin{proof}
Bannister \emph{et al.}~\cite{BanCheDev-GD-2013} showed that $S_n(213)$ has a superpattern $\mu_n$ of size $n^2 / 4 + n + O(1)$. For a given $n$, let $\rho_n$ be a permutation whose first element is its smallest and whose remaining elements are order-isomorphic to $\mu_{n-1}$.
 The point set $\plot(\rho_n)$ provides the desired universal point set.
\end{proof}

We remark that another result of Bannister \emph{et al.}, on the superpatterns of proper subclasses of the $213$-avoiding permutations, implies that for every constant $s$ there exist universal point sets of size $O(n\log^{O(1)} n)$ for dominance drawings of the trees of Strahler number~$\le s$. We omit the details.

\section{Conclusion}

We have extended the connection between dominance drawing and permutations, and between universal point sets and superpatterns, initially discovered by Bannister \emph{et al.}~\cite{BanCheDev-GD-2013}. Using this connection, we have found new small universal sets for dominance drawings of width-two Hasse diagrams, $st$-outerplanar graphs, and directed trees, and new small superpatterns for $321$-avoiding permutations and several of their subclasses.
Interestingly, although the $213$-avoiding and $321$-avoiding permutations are equinumerous~\cite{SimSch-EJC-85}, the new superpatterns we have found for the $321$-avoiding permutations are significantly smaller than the superpatterns found by Bannister \emph{et al.} for the $213$-avoiding permutations.

Our investigation points to several additional questions about superpatterns and universal point sets that we have not yet been able to answer, and that we leave for future research:
\begin{itemize}
\item Bannister \emph{et al.} showed that every proper subclass of the $213$-avoiding permutations has superpatterns of near-linear size. Does a similar result hold for the $321$-avoiding permutations?
\item An important subclass of the $st$-planar graphs, for which we have been unable to provide improved universal point sets, are the (transitively reduced) directed series-parallel graphs. The dominance drawings of these graphs correspond to the separable permutations, permutations with $2413$ and $3142$ as forbidden patterns. Do these permutations have superpatterns of size smaller than the $n^2/2-O(n)$ bound known for the superpatterns of all permutations?
\item Can our $O(n^{3/2})$ bound on the size of superpatterns for $321$-avoiding permutations be extended to a bound of $O(n^{2-1/(k-1)})$ on $k,k-1,\dots,3,2,1$-avoiding permutations? If so we would also obtain similarly sized universal point sets for dominance drawings of Hasse diagrams of the partial orders of dimension two and width~$k$.
\item Can we prove any nontrivial lower bounds on the size of superpatterns for $321$-avoiding permutations or their subclasses?
\end{itemize}

{\raggedright
\bibliographystyle{abuser}
\bibliography{paper}}
\end{document}